\documentclass[runningheads]{llncs}
\usepackage[T1]{fontenc}   
\usepackage{amsmath}
\usepackage{graphicx}
\newtheorem{fact}{Fact}
\newcommand{\junk}[1]{}

\begin{document}





\title{Approximate all-pairs Hamming distances and 0-1 matrix multiplication}
\author{
Miros{\l}aw Kowaluk
\inst{1}
\and
Andrzej Lingas
\inst{2}
\and
  Mia Persson
  \inst{3}}
\institute{
  Institute of Informatics, University of Warsaw, Warsaw, Poland.
  \texttt{kowaluk@mimuw.edu.pl}
\and
  Department of Computer Science, Lund University,
Lund, Sweden.
\texttt{Andrzej.Lingas@cs.lth.se}
\and
Department of Computer Science and Media Technology, Malm\"o University, Malm\"o, Sweden.
\texttt{Mia.Persson@mau.se}
}
\pagestyle{plain}

\maketitle

\begin{abstract}
  Arslan showed that computing all-pairs Hamming distances is easily
  reducible to arithmetic 0-1 matrix multiplication (IPL 2018).  We
  provide a reverse, linear-time reduction of arithmetic 0-1 matrix
  multiplication to computing all-pairs distances in a Hamming space.
  On the other hand, we present a fast randomized algorithm for
  approximate all-pairs distances in a Hamming space. By combining it
  with our reduction, we obtain also a fast randomized algorithm for
  approximate 0-1 matrix multiplication. Next, we present an
  output-sensitive randomized algorithm for a minimum spanning tree of
  a set of points in a generalized Hamming space, the lower is the
  cost of the minimum spanning tree the faster is our
  algorithm. Finally, we provide $(2+\epsilon)$- approximation
  algorithms for the $\ell$-center clustering and minimum-diameter
  $\ell$-clustering problems in a Hamming space $\{0,1\}^d$ that are
  substantially faster than the known $2$-approximation ones when both
  $\ell$ and $d$ are super-logarithmic.
\end{abstract}

\section{Introduction}
Arslan observes in \cite{A18} that computing all-pairs 
distances in Hamming spaces $\{0,1\}^d$ and generalized Hamming spaces
$\Sigma^d,$ where $\Sigma$ is a finite alphabet, forms a frequent
major step in hierarchical clustering \cite{AL15} and in
phylogenetic-tree construction \cite{AB07}.  He also lists several
examples of applications of hierarchical clustering based on Hamming
distances in medical sciences \cite{A18}.  Once pairwise Hamming
distances between input points in a generalized Hamming space are known, a minimum
spanning tree of the points can be easily computed. The computation of
exact or approximate minimum spanning trees in
Hamming spaces and generalized Hamming spaces of high
dimension is widely used technique for clustering data, especially in
machine learning \cite{GZJ06,LRN09}.

In \cite{A18}, Arslan reduced the problem of computing all-pairs
Hamming distances in a generalized Hamming space to arithmetic matrix multiplication.
In particular, he obtained an $O(|\Sigma|n^{\omega})$-time algorithm for this problem
when there are $n$ input points in $\Sigma^d$ and $\omega$
denotes the exponent of the fast arithmetic matrix multiplication.
The currently best known upper bound on $\omega$ is
$2.3714$ \cite{ADV25,VXZ24}.

\junk{Clustering is nowadays a standard toll in the data
  analysis in computational biology/medical sciences, computer vision,
  and machine learning.}
One of the most popular variants
  of clustering are the $\ell$-center problem and the related
  $\ell$ minimum-diameter problem in metric spaces
  (among other things in Hamming spaces).
  Given a finite set $P$ of
points in a metric space,
the first problem asks for 
finding a set of $\ell$ points in the metric space, called
{\em centers}, such that the maximum distance of a point
in $P$ to its nearest center is minimized.
The second problem asks for partitioning
the input point set $P$ into $\ell$ clusters such that the maximum
of cluster diameters is minimized.
  They are known to be NP-hard and even
  NP-hard to approximate within $2-\epsilon$ for any
  constant $\epsilon > 0$ \cite{Gon85}.
  Gonzalez provided a simple $2$-approximation method
 for $\ell$-center clustering that yields also a $2$-approximation
 for minimum-diameter $\ell$-clustering \cite{Gon85}.
 In case of a  $d$-dimensional space, his method takes
 $O(nd\ell)$ time, where $n$ is the number of input points.
 For low dimensional Euclidean spaces of bounded
 dimension, and more generally, for metric spaces of bounded doubling
 dimension, there exist faster $2$-approximation algorithms for the
 $\ell$-center problem with hidden exponential dependence on the
 dimension in their running time, see \cite{FG88} and \cite{HM06},
 respectively. In some of the aforementioned applications of both
 $\ell$-clustering problems,
 massive datasets in a high dimensional metric space combined
 with a large value of the parameter $\ell$ may occur. In such situations,
 neither the $O(nd\ell)$-time method
 nor the ones of time complexity exponentially dependent
 on $d$ are sufficiently efficient.

The arithmetic matrix product of two 0-1 matrices is closely related
to the Boolean matrix product of the corresponding matrices and for
$n\times n$ matrices can be computed in $O(n^{2.372})$ time
\cite{ADV25,VXZ24}.  It is a basic tool in science and engineering
(e.g., in machine learning \cite{AMB24}).
Unfortunately, no truly subcubic practical algorithms for the
arithmetic product of $n\times n$ 0-1 matrices are known.  Therefore,
in some cases, a faster approximate arithmetic matrix multiplication
that enables for the identification of largest entries in the product
matrix can be more useful \cite{CL99,P13}.
Among other things, it can be also used to provide a fast 
approximation of the number of: the so called witnesses for the Boolean
product of two Boolean matrices \cite{GKL09}, triangles in a graph,
or more generally, subgraphs isomorphic to a small
pattern graph \cite{FKL15} etc.
There is a number of results on
approximate arithmetic matrix multiplication, where the quality of approximation
is 
expressed in terms of the Frobenius matrix norm $||\ ||_F$
(i.e., the square root of
the sum of the squares of the entries of the matrix) \cite{CL99,P13}.

Cohen and Lewis \cite{CL99}, followed by Drineas {\em et al.} \cite{DKM06},
used random sampling to approximate arithmetic matrix product. 
  Both papers provide an approximation $D$
  of the matrix product $AB$ of two $n\times n$ matrices $A$ and $B$ 
  guaranteeing $||AB-D||_F=O(||AB||_F/\sqrt c)$, for a parameter $c>1$ (see
  also \cite{P13}).
  In \cite{CL99} the matrices $A,\ B$  are required
  to be nonnegative while in \cite{DKM06} they are arbitrary.
  The approximation in \cite{DKM06} takes $O(n^2c)$ time.
  Drineas et al. \cite{DKM06} also provided bounds on
  the differences $|AB_{ij}-D_{ij}|$ on the entry level.
  Unfortunately, the best of these
  bounds is $\Omega(Q^2n/\sqrt c)$, where $Q$ is the maximum value
  of an entry in $A$ and $B.$
  Sarl\'os \cite{S06} 
  obtained
  the same Frobenius norm guarantees,
also in $O(n^2c)$ time. However, he derived stronger individual
upper bounds on the additive error of each entry $D_{ij}$ of the
approximation matrix $D$ of the form $O(||A_{i*}||_2||B_{*j}||_2/\sqrt c)$,
where $A_{i*}$ and $B_{*j}$ stand for the $i$-th row of $A$
and $j$-th column of $B,$ respectively. These bounds hold with
high probability.
More recently Pagh \cite{P13} presented a
randomized approximation $\tilde{O}(n(n+c))$-time
algorithm for the arithmetic product
of $n\times n$ matrices $A$ and $B$
such that each entry of the approximate
matrix product differs from the correct one at most by
$||AB||_F/\sqrt c$. Pagh's algorithm first compresses
the matrix product to a product of two polynomials and then applies the
fast Fourier transform  to multiply the polynomials.
Subsequently, Kutzkov \cite{K13} presented analogous deterministic algorithms
employing different techniques. For approximation results related to sparse
arithmetic matrix products, see \cite{IS09,P13}.
\junk{
E.g., Pagh presented an approximation $\tilde{O}(n(n+b))$-time
algorithm for the arithmetic multiplication
of two $n\times n$ matrices such that each entry of the approximate
matrix product differs from the correct one at most by the value
of the Frobenius norm for the correct product divided by
the square root of $b,$ where $b$
is a positive integer parameter \cite{P13}.}

\subsection{Our contributions}

Our first result is a linear-time (with respect to input
and output size) reduction of arithmetic 0-1 matrix
multiplication to computing all-pairs Hamming distances in a
Hamming space.  Combined with the reverse
reduction from \cite{A18},
it implies that these two problems have the same asymptotic time
complexity.
Our next result is a fast randomized algorithm for
approximate all-pairs Hamming distances in a
Hamming space. To derive it, we show a lemma on randomized
dimension reduction in Hamming spaces in the spirit of
\cite{AIR18,KOR00}.
For $n$
points in the $d$-hypercube, our algorithm
approximates all the pairwise Hamming
distances between them with high probability within $1+\delta$ in time
$O((\log N /\delta^2) (N\log_{1+\delta}d +n^2\log \log_{1+\delta}d))$, where
$N=nd$ is the input size. E.g., when $d=n$ and $\delta =\Omega (1),$
it takes $\tilde{O}(n^2)$ time while the algorithm for the exact
pairwise distances due to Arslan requires
$O(n^{\omega})$ time
\cite{A18}.
By combining our approximation algorithm with our reduction, we
obtain also a fast randomized algorithm for approximate 0-1 matrix
multiplication. With high probability, the approximation of the inner
product of any row of the first matrix and any column of the second
matrix differs at most by an $\epsilon$ fraction of the minimum of the
Hamming distance between the row and the column and the complement of
the distance. If the minimum is within a constant factor of the exact
value of the inner product, the approximation is tight.
Then, we present an output-sensitive randomized algorithm for a minimum
spanning tree of a set of $n$ points in the generalized Hamming space
$\Sigma^d,$ where $\Sigma$ is an alphabet of
$O(1)$ size. With high probability, it runs
in $\tilde{O}(n(d+n+M))$ time, where $M$ is the (Hamming) cost of a
minimum spanning tree. Interestingly, a major preliminary step
in our algorithm is a construction of an approximate
minimum spanning tree of the points based on our fast randomized
algorithm for approximate all-pairs Hamming distances.
Finally, we show that our lemma on randomized dimension reduction
in Hamming spaces yields also $(2+\epsilon)$- approximation algorithms
for the $\ell$-center clustering
and minimum-diameter $\ell$-clustering problems in a Hamming space
$\{0,1\}^d$ that are substantially faster than the known $2$-approximation
ones when both $\ell$ and $d$ are super-logarithmic.

\subsection{Techniques}
Our fast randomized algorithm for approximate all-pairs 
distances in a Hamming space is based on a  variant
of randomized
dimension reduction in Hamming spaces
in the spirit of
\cite{AIR18,KOR00}.  The crucial
step in our output-sensitive randomized algorithm for a minimum
spanning tree in a generalized Hamming space is the computation of
all-pairs Hamming distances between the input points. It is
implemented by an adaptation and generalization of the known method of multiplying
Boolean or arithmetic square matrices in time dependent on the cost of
an approximate minimum spanning tree of the rows of the first matrix or the cost of
a minimum spanning tree of the columns of the second matrix, where the
rows and columns are treated as points in a hypercube
\cite{BL,FJLL18,GL03}.

\subsection {Paper organization}
 The next section contains basic
 definitions. Section 3 demonstrates our
 reduction of 0-1 matrix multiplication to
 computing all-pairs Hamming distances
 between the
rows of the first matrix and the columns of the second matrix.
Section 4 presents a lemma on randomized dimension
reduction in Hamming spaces and our
 fast randomized algorithm for approximate all-pairs
 Hamming distances based on the reduction.
 Section 5 presents our
 fast randomized algorithm for approximate 0-1
 matrix multiplication. Section 6 is devoted to
 our output-sensitive algorithm
 for a minimum spanning tree in a
 generalized Hamming space. Section 7 presents
 our randomized algorithm for approximate $\ell$-clustering
 problems.
We conclude with final remarks.

\section{Preliminaries}

For a positive integer $r$, $[r]$ stands for the set of positive integers not exceeding $r.$

The transpose of a matrix $D$ is denoted by $D^{\top}.$
If the entries of $D$ are in $\{0,1\}$ then
$D$ is a 0-1 matrix while if they are in a finite alphabet $\Sigma,$
$D$ is a $\Sigma$ matrix.

The symbol $\omega$ denotes the smallest real number such that two $n\times n$
matrices can be multiplied using $O(n^{\omega +\epsilon})$
operations over the field of reals, for all $\epsilon > 0.$

The {\em Hamming distance} between two points $a,\ b$ (vectors) in
$\{0,1\}^d$, and more generally, in $\Sigma^d$, is the number of the
coordinates in which the two points differ.  Alternatively, in case of
$\{0,1\}^d$ it can be defined as the distance between $a$ and $b$ in
the $L_1$ metric over $\{0,1\}^d.$ It is denoted by $\mathrm{ham}(a,b).$

For a positive real $\delta,$ an estimation of the Hamming distance
between two points $a,\ b\in \{0,1\}^d$ whose value
is in $[\mathrm{ham}(a,b)/(1+\delta), (1+\delta)\mathrm{ham}(a,b)]$
\junk{differs from
$\mathrm{ham}(a,b)$ at most by $\delta \cdot \mathrm{ham}(a,b)$}
is called a
{\em $\delta$-approximation} of $\mathrm{ham}(a,b).$
For a finite set $S$ of points in $\{0,1\}^d$,
a {\em $\delta$-approximate nearest neighbor}
of a point $a\in \{0,1\}^d$ in $S$ is a point $b\in S\setminus \{a\}$
such that $\mathrm{ham}(a,b)\le (1+\delta)\mathrm{ham}(a,c)$
for any point $c\in S\setminus \{ a\}.$

An event is said to hold {\em with high probability} (w.h.p. for short) in
terms of a parameter $N$ related to the input size if it holds with
probability at least $1-\frac 1 {N^{\alpha}}$, where $\alpha $ is any
constant not less than $1.$

\section{All-pairs Hamming distances versus 0-1 matrix multiplication}

Arslan studied the problem of computing the Hamming distances between
all pairs of points belonging to two point sets 
in a generalized Hamming space $\Sigma^d,$
where $\Sigma$ is a finite alphabet
\cite{A18}. He provided a fast algorithm for this
problem based on a reduction to arithmetic matrix multiplication 
\cite{A18}.The following theorem shows also a
reverse reduction. In the theorem, the two sets of points in the generalized  Hamming space are represented
by two matrices, where the vectors of coordinates of the points form
the rows or the columns, respectively.

\begin{theorem}\label{theo: red}
  Let $A$ and $B$ be two $\Sigma$ matrices of sizes $p\times q$ and
  $q\times r,$ respectively.
  \begin{enumerate}
  \item
   (\cite{A18})  The problem of computing the Hamming
  distance between each row of $A$ and each column of $B$ can be
  reduced in $O(|\Sigma|pr)$ time to that of computing $O(|\Sigma|)$ 
  arithmetic products of 0-1 matrices of sizes $p\times q$ and
  $q\times r$, respectively.
  
\item
  Conversely in case $\Sigma=\{0,1\}$, given the Hamming distances
  between each row of $A$ and each column of $B$,
  the arithmetic  matrix product of $A$ and $B$ can be computed
  in $O(pq+qr+pr)$ time.
  \end{enumerate}
\end{theorem}
\begin{proof}
  For the proof of the first part see \cite{A18}.  In the special case
  $\Sigma=\{0,1\}$ and $p,q,r=n$, it can be also found in the proof
  of Theorem 1 in \cite{L24}.  Here, we just give a
  similar short proof for the case $\Sigma=\{0,1\}$.  For a 0-1
  matrix $D$, let $\bar{D}$ denote the matrix obtained from $D$ by
  flipping the $1$ entries to $0$ and {\sl vice versa}.  Next, let
  $C,\ \bar{C}$ stand for the arithmetic matrix products of
  $A,\ B,$ and $\bar{A},\ \bar{B},$ respectively.  Then, the Hamming
  distance between the $i$-th row of $A$ and the $j$-th column of $B$
  is easily seen to be equal to $q-C_{ij}-\bar{C}_{ij}.$

  Conversely in case $\Sigma=\{0,1\}$,
  suppose that the Hamming distances \\
  $\mathrm{ham}(A_{i*},B_{*j})$
  between each row $A_{i*}$
  of $A$ and each column $B_{*j}$ of $B$ are
  given.
  \junk{For $k\in \{0,1\}$, let $A_{i*}^k$ and $B_{*j}^k$ denote the
  number of occurrences of $k$ in $A_{i*}$ and $B_{*j}$,
  respectively. Similarly, for $k\in \{0,1\}$, let $XA_{i*}^{k,j}$ and
  $XB_{*j}^{k,i}$
  denote the number of occurrences of $k$ in $A_{i*}$ and $B_{*j}$,
  respectively, that occur on the same positions as occurrences of
  $k+1 \mod 2$ in $B_{*j}$ and $A_{i*}$, respectively.
  Finally, let $x$ denote the unknown value of the inner arithmetic
  product $C_{ij}$ of $A_{i*}$ and $B_{*j}.$
  In particular, the following equations hold.

  $$A_{i*}^0=XA_{i*}^{0,j} + (q-\mathrm{ham}(A_{i*},B_{*j})) -x$$
  $$B_{*j}^1=XB_{*j}^{1,i}+x$$
  $$XA_{i*}^{0,j}=XB_{*j}^{1,i}$$

  They reduce to two following equations
  with two variables.

  $$A_{i*}^0=XA_{i*}^{0,j} + (q-\mathrm{ham}(A_{i*},B_{*j})) -x$$
  $$B_{*j}^1=XA_{i*}^{0,j}+x$$

  This yields

  $$x=\frac {q-\mathrm{ham}(A_{i*},B_{*j})-A_{i*}^0+B_{*j}^1}2
  =\frac {A_{i*}^1 + B_{*j}^1-\mathrm{ham}(A_{i*},B_{*j})}2$$}
For a 0-1 vector $(v_1\dots,v_k)$,
let $(v_1,\dots,v_k)^1$ denote the number of
    1s in the vector, i.e., $\sum_{\ell \in [k]}v_{\ell}.$
    It is easy to verify that for
    two one-dimensional 0-1 vectors $(a)$ and $(b)$,
    the equality $ab=\frac {(a)^1+(b)^1-\mathrm{ham}((a)),((b))}2$
    holds. Hence, we can compute the inner product
    $C_{ij}$ of $A_{i*}$ and $B_{*j}$ as follows:
    $$C_{ij}=\sum_{\ell \in [q]}A_{i\ell}B_{{\ell}j}=$$
    $$\sum_{\ell \in [q]}\frac {(A_{i\ell})^1+
      (B_{{\ell}j})^1-\mathrm{ham}((A_{i\ell}),(B_{{\ell} j}))}2
      =\frac {A_{i*}^1 + B_{*j}^1-\mathrm{ham}(A_{i*},B_{*j})}2$$

    Thus, it is sufficient to precompute the
    numbers of $1$s  in each row of $A$
    and each column of $B$ in $O(pq+qr)$
    time in order to compute the entries
    $C_{ij}$ of the arithmetic matrix product $C$ of $A$ and $B$
    on the basis of $\mathrm{ham}(A_{i*},B_{*j})$ in $O(pr)$ time.
    \qed
\end{proof}

We can compute an exact minimum spanning tree of $n$ points
$p_1,\ p_2,\dots,p_n$ in the generalized Hamming space $\Sigma^d$ by running a
standard linear-time (Prim's) algorithm for a minimum spanning tree on the
clique on $[n]$, where the weight of the edge $(i,j)$ is set to the
Hamming distance between the points $p_i$ and $p_j$.  To obtain the
Hamming distances, we just form a matrix $P$ where the $i$-th row is
the vector of coordinates of the point $p_i$ for $i\in [n],$ and apply
Theorem \ref{theo: red} (1) to the matrices $P$ and $P^{\top}.$

\begin{corollary}\label{cor: gmst}
A minimum spanning tree of  $n$ points
$p_1,\ p_2,\dots,p_n$ in the generalized Hamming space $\Sigma^d$
can be constructed in $O(|\Sigma|T(n,d))$ time, where
$T(n,d)$ is the time required to multiply two
rectangular 0-1 matrices of sizes $n\times d$ and $d\times n,$
respectively.
\end{corollary}

\section{Approximate all-pairs Hamming distances via dimension
reduction}

The following lemma and remark enable an efficient randomized dimension
reduction for the purpose of computing all-pairs approximate Hamming
distances for a set of $n$ input points (vectors) in $\{0,1\}^d$.
\junk{
\begin{fact}\label{fact: R}\cite{AIR18,KOR00} (Lemma 2.3 in \cite{AIR18})
Fix the error parameter $\epsilon \in (0, 1/2),$ dimensions $k, d \ge 1,$
and scale (threshold) $t\in [1, d].$
For any $k \ge 1,$ there exists a randomized map $f : \{0,1\}^d \rightarrow \{0,1\}^k$ and an absolute
constant $C > 0,$
satisfying the following conditions for any fixed $x, y \in \{0,1\}^d :$
\begin{description}
\item
  if $\mathrm{ham}(x,y)\le t,$ then $Pr[\mathrm{ham}(f (x),f (y))\le k/2] \ge  1-e^{-C\epsilon^2 k}$;\\
\item
  if  $\mathrm{ham}(x,y)\ge (1 + \epsilon)t,$  then $Pr[\mathrm{ham}(f (x),f (y))> (1+\epsilon/2)k/2]\ge  1-e^{-C\epsilon^2 k}$.
\end{description}
\end{fact}

\begin{remark}\cite{AIR18}\label{rem: 1}
  The map $f$ can be constructed in $O(dk)$ time
via a random projection over
  $GF (2)$ by taking $f (x) = Ax,$
  where $A$ is a $k\times d$ matrix,
  with each entry being $1$ with some fixed probability $\phi$
  and $0$ otherwise. The probability $\phi$ depends solely on $t$.
  For  $k = O(\log N/\epsilon^2)$, where $N$ is related
  to the input size, 
  the conditions in Fact \ref{fact: R} hold with high probability
  (in terms of $N$).
\end{remark}}

\begin{lemma}\label{lem: R}(in the spirit  of \cite{AIR18,KOR00}, cf. Lemma 2.3 in \cite{AIR18})
Fix the error parameter $\epsilon \in (0, 1/2),$ dimensions $k, d \ge 1,$
and scale (threshold) $t\in [1, d].$
For any $k \ge 1,$ there exist a randomized map $f : \{0,1\}^d \rightarrow \{0,1\}^k$, an absolute constant $C,$ and a
constant $C_1 > 0$ (dependent only on $t$) such
that for any $\epsilon > 0$ there is a constant $C_2>C_1$
(dependent only on $t$ and $\epsilon$) such that
$C_2-\epsilon/30 > C_1+\epsilon/30$ and
the following conditions hold for any $x, y \in \{0,1\}^d$:
\begin{description}
\item
  if $\mathrm{ham}(x,y)\le t,$ then
  $Pr[\mathrm{ham}(f (x),f (y))\le (C_1+ \frac {\epsilon} {30})k] 
\ge  1-e^{ -C\epsilon^2 k}$;
\item
  if  $\mathrm{ham}(x,y)\ge (1 + \epsilon)t,$  
then $Pr[\mathrm{ham}(f (x),f (y))> (C_2-\frac {\epsilon} {30})k]
  \ge  1-e^{-C \epsilon^2 k}.$
\end{description}
\end{lemma}
\begin{proof}
  See Appendix.
  \end{proof}
\begin{remark}\label{rem: 1}
 The map $f$ can be constructed in $O(dk)$ time
according to the proof
of Lemma  \ref{lem: R} as follows:
generate a random $k\times d$ 0-1 matrix $F,$
  where each entry is set to $1$
  with probability $\frac 1 {4t}$ and
  for $x\in \{ 0,1\}^d,$ let $f(x)=Fx^{\top}$ $\mod 2.$
  For  $k = O(\log N/\epsilon^2)$ (e.g., $N$ can be related
  to the input size),
  the conditions in Lemma \ref{lem: R} hold with high probability
  (in terms of $N$).
\end{remark}

The following procedure computes a $\delta$-approximate Hamming
distance between each row of the first input 0-1 matrix $A$ of size
$p\times q$ and each column of the second input 0-1 matrix $B$ of size
$q\times r$.  It relies on Lemma \ref{lem: R} and Remark \ref{rem: 1}.
First, the pairs consisting of a row of $A$ and a column of $B$
that are identical (i.e., within $0$ Hamming distance) are identified
by sorting the rows of $A$ and the columns of $B$ jointly.
Next, random 0-1 matrices of size $k\times q$ are generated for
thresholds $t=(1+\delta)^{\ell},$ $\ell =0,\dots, \lceil \log_{1+\delta}
q \rceil ,$ with $k=O(\log N/\delta^2)$ and $N=pq+qr$, according to Remark
\ref{rem: 1}. Then, for each row of $A$ and each column of $B$ that
are not identical,
the approximate Hamming distance between the row and the column
is set to the smallest threshold $t$ for which the random projections
of the row and the column obtained by multiplication of the random
matrix corresponding to $t$ with the vectors of their coordinates
mod 2 are within the Hamming distance $(C_1+\frac {\epsilon} {30})k$
 ($C_1$ is a constant given in Lemma \ref{lem: R}).

\bigskip

    \noindent
        {\bf procedure} $APPAPHAM (A,B, \delta )$
        \par
        \noindent
            {\em Input}: Two 0-1 matrices $A$ and $B$
            of sizes $p\times q$ and $q\times r,$
            respectively, and a real  $\delta \in (0,\frac 12).$
            \par
            \noindent
                {\em Output}: A $p\times r$ matrix $W$, where
                for $1\le i\le p$ and $1\le j\le r,$ $W_{ij}$
                is an approximation of the Hamming distance between
                the $i$-th row $A_{i*}$ of $A$  and the $j$-th column
                $B_{*j}$ of $B.$
                \begin{enumerate}

                \item Sort the rows of $A$ and the columns of $B$
                  jointly as $q$-bit
                  binary numbers. For each maximal block $b$ of
                  consecutive identical rows of $A$
                    and/or columns of $B$, for each $A_{i*}\in b$ and
                    each $B_{*j}\in b$,
                    set $W_{ij}$ to $0.$
                  \item
                  Set $\epsilon$ to $\delta$
                  and $k$ to $O(\log N /\delta^2)$,
                  where $N=pq+qr$.
                 \item
                 Generate random $k\times q$     0-1 matrices 
                $F_t$ for the thresholds $t=1,\ (1+\delta),\dots, (1+\delta)^{\lceil \log _{1+\delta} q\rceil },$
                  defining the functions $f_t$ by $f_t(x)=F_tx^{\top}\mod 2$
                  (see Remark \ref{rem: 1}).
                \item
                  Compute the values of all the functions $f_t$  for each
                  row of $A$ and each column of $B.$,
                  i.e.,
                  for $1\le i\le p$ and $1\le j \le r,$
                  compute $F_tA_{i*}^{\top}\mod 2$ and $F_tB_{*j}\mod 2$ for
                  $t=1,\ (1+\delta),\dots, (1+\delta)^{\lceil  \log _{1+\delta}q \rceil}.$
                 \item
                   For $1\le i\le p$ and $1\le j \le r,$
                   if $W_{ij}$ has not been defined (i.e., set to $0$) in (1) then
                   set $W_{ij}$  to the smallest $t\in \{
1,\ (1+\delta),\dots, (1+\delta)^{\lceil \log _{1+\delta}q\rceil } \}$
such that $\mathrm{ham}(f_t(A_{i*}),f_t(B_{*j}))\le (C_1+\frac {\epsilon}
{30})k.$ (see Lemma \ref{lem: R})
                                \end{enumerate}

                \begin{lemma}\label{lem: first}
                                    $APPAPHAM(A,B,\delta)$ can be implemented
                                    in  time \\
                                    $O((\log N /\delta^2 )(N\log_{1+\delta}q+pr\log \log_{1+\delta}q))$,
                  where $N$ stands for the input size $pq+qr.$
                \end{lemma}
                \begin{proof}
                  Step 1 can be easily implemented in $O(pq+qr+pr)$ time
                  by radix sort and scanning the sorted list.
                  Steps 2, 3 take $O(kq\log_{1+\delta} q)$ time.
                  Step 4 can be implemented in $O((pq+qr)k\log_{1+\delta}q)$
                  time. Finally, Step 5 takes 
                  $O(pr k\log \log_{1+\delta}q)$ time
                  by using binary search.
                  It remains to recall that $k=O(\log N/\delta^2).$
                  \qed
                \end{proof}
                
\begin{lemma}\label{lem: ineq}
  Assume the notation from the procedure $APPAPHAM(A,B,\delta).$ For
  $1\le i\le p$ and $1\le j \le r,$
the following inequalities  hold  w.h.p.:

$$W_{ij}/(1+\delta)\le \mathrm{ham}(A_{i*},B_{*j}),$$
$$ \mathrm{ham}(A_{i*},B_{*j})\le (1+\delta)W_{ij}.$$
\end{lemma}
\begin{proof}
Suppose that $W_{ij}\ge (1+\delta)$ and $\mathrm{ham}(A_{i*},B_{*j})< W_{ij}/(1+\delta)$ hold. Note that $W_{ij}=(1+\delta)^k$ for some positive integer $k.$
Then, the inequality $\mathrm{ham}(f_t(A_{i*}),f_t(B_{*j}))
\le (C_1+\frac {\delta} {30})k$
should hold already for $t=W_{ij}/(1+\delta)$ w.h.p. by the first condition  in 
Lemma \ref{lem: R}
with $\epsilon=\delta$ and  sufficiently large $k=O(\log N/ \delta^2).$
We obtain a contradiction with the choice of $W_{ij}$ w.h.p.
Suppose in turn that $W_{ij} < (1+\delta)$, then $W_{ij}$ is equal to
$1$ or $0$ by its minimality w.h.p.
If  $\mathrm{ham}(A_{i*},B_{*j})= 1$,
the inequality $W_{ij}/(1+\delta)\le  \mathrm{ham}(A_{i*},B_{*j})$ trivially holds.
Otherwise, we have $\mathrm{ham}(A_{i*},B_{*j})=0$ and consequently $W_{ij}=0$
by Step 1 in $APPAPHAM(A,B,\delta).$ This completes the proof of
the first inequality.

To prove the second inequality, suppose 
that $\mathrm{ham}(A_{i*},B_{*j})> (1+\delta)W_{i j}$ holds.
Then, $\mathrm{ham}(f_t(A_{i*}),f_t(B_{*j}))\ge (C_2-\frac {\delta} {30})k$ holds
for $t=W_{ij}$ w.h.p.
by the second condition  in Lemma \ref{lem: R} with  $\epsilon=\delta$.
The inequality $\mathrm{ham}(f_t(A_{i*}),f_t(B_{*j}))> (C_1+\frac {\delta} {30})k$
follows from Lemma \ref{lem: R}.
We obtain again a contradiction with the choice of $W_{ij}$ w.h.p.
\qed
\end{proof}

\begin{theorem}\label{theo: approxham} 
  Let $\delta \in(0,\frac 12)$, and let
  $A$ and $B$ be two 0-1 matrices of sizes $p\times q$ and
  $q\times r,$ respectively.
  W.h.p., for $1\le i\le p$ and $1\le j \le r$,
  $\delta$-approximations 
  of $\mathrm{ham}(A_{i*},\ B_{*j})$
  can be computed 
  in  time $O((\log N /\delta^2 )(N\log_{1+\delta}q+pr\log \log_{1+\delta}q))$,
where $N$ stands for the input size $pq+qr.$
\end{theorem}
\begin{proof}
  We run $APPAPHAM(A,B,\delta)$. Lemmata \ref{lem: first} and \ref{lem: ineq} yield the theorem.
  \qed
  \end{proof}

We can approximate a minimum spanning tree of $n$ points
$p_1,\ p_2,\dots,p_n$ in the Hamming space $\{0,1\}^d$ by running
a standard linear-time (Prim's) algorithm for a minimum spanning
tree on the clique
on $[n]$, where the weight of the edge $(i,j)$ is
set to the $\delta$-approximation $W_{ij}$ of the Hamming distance
between the points $p_i$ and $p_j$ given by $APPAPHAM.$
To obtain the approximations, we form the matrix
$P$ where the $i$-th row is the vector of coordinates
of the point $p_i$ for $i\in [n],$ and call
  $APPAPHAM(P,P^{\top},\delta).$ Let $T_{\delta}$ denote
  the minimum spanning tree of the weighted
  clique on $[n]$ produced by Prim's algorithm. Next, let
$T'_{\delta}$ be a spanning tree of
    the points $p_1,\ p_2,\dots, p_n$ obtained
    by substituting $p_i$ for $i$ in $T_{\delta},$
    for $i\in [n].$
    
  \begin{theorem}\label{theo:  approxmst}
   W.h.p., the cost of  $T'_{\delta}$  in
    the Hamming space $\{0,1\}^d$ is within
    the multiplicative factor $(1+\delta)^2$
    of the cost of a minimum spanning tree of
    the points $p_1,...,p_n$ in $\{0,1\}^d$.
    The tree $T_{\delta}$
    can be constructed in time \\
    $O((\log (nd)/\delta^2)
    (nd\log_{1+\delta}d +n^2\log \log_{1+\delta}d ))$.
    \end{theorem}
\begin{proof}
  Let $U$ be  a minimum spanning tree 
  of the input points $p_1,...,p_n$.
  By Lemma \ref{lem: ineq} and the optimality of the
  minimum spanning tree $T_{\delta}$
  on the clique on $[n]$ with the
  edge weights $W_{ij}$,
  we obtain the following chain of inequalities w.h.p.:
$$\sum_{(p_i,p_j)\in T'_{\delta}}\mathrm{ham}(p_i,p_j)\le \sum_{(i,j)\in T_{\delta}}\mathrm{ham}(p_i,p_j)\le
(1+\delta) \sum_{(i,j)
  \in T_{\delta}}  W_{ij} $$
$$\le (1+\delta)\sum_
{(p_{\ell} ,p_m)\in U}  W_{\ell m}\le (1+\delta)^2\sum_{(p_{\ell},p_m)\in U} \mathrm{ham}(p_{\ell},p_m).$$

The upper time bound follows from Lemma \ref{lem: first} and the fact
that a minimum spanning tree of the weighted clique on $[n]$ can be found
in $O(n^2)$ time.
\qed
\end{proof}

Theorem \ref{theo: approxmst} can be compared with the following
fact restricted to the $L_1$ metric in $\{0,1\}^d,$ i.e.,
the Hamming metric in $\{0,1\}^d.$

\begin{fact}\label{fact: mst}
  (Section 3 in \cite{IM}). For $\delta>0,$ a $(1+\delta)$-approximation of a
minimum spanning tree for a set of $n$ points in $R^d$
with integer coordinates in $O(1)$ under the
$L_1$ or $L_2$ metric can be computed by a Monte Carlo
algorithm in $O(dn^{1+1/(1+\delta )})$ time.
\end{fact}

We obtain also the following weak generalization of Theorem \ref{theo:  approxmst} to include
the generalized Hamming space $\Sigma^d,$ where $\Sigma$ is a finite alphabet.

\begin{corollary} \label{cor:g gapproxmst}
  Suppose $\Sigma$ is an alphabet of  $O(1)$ size.
  An approximate minimum  spanning tree of $n$ points $p_1,\dots, p_n$ in the generalized
  Hamming space $\Sigma^d$ whose cost is within $O(\log_2|\Sigma| +1)$ of the cost
  of a minimum spanning tree of the points can be constructed
  in  time $O(\log (nd)
    (nd\log d +n^2\log \log d ))$.
 \end{corollary}
\begin{proof}
  We may assume without loss of generality
  that $|\Sigma|>1.$
  Let $g$ be an embedding of  $\Sigma^d$ into $\{ 0, 1\}^{d\lceil \log_2|\Sigma|\rceil }$ 
  obtained by replacing the value of each coordinate of a point in $\Sigma^d$
  by an $\lceil \log_2|\Sigma|\rceil $-digit  binary encoding of the value.
  Note that for any $a,\ b \in \Sigma^d,$
  $\mathrm{ham}(a,b)\le \mathrm{ham}(g(a),g(b))\le \lceil \log_2|\Sigma|\rceil  \mathrm{ham}(a,b).$
  Hence, it is sufficient to apply Theorem \ref{theo: approxmst}  to the images $g(p_1),\dots,g(p_n)$
  of the input points with $\delta$ set, say, to $1.$
  \qed
\end{proof}


In a similar fashion, we can solve the problem of computing for each
point in a set $S$ of $n$ points in $\{0,1\}^d$ a
$\delta$-approximate nearest neighbor in $S$ by using our method for
all-pairs $\delta$-approximate Hamming distances given in Theorem \ref{theo: approxham}.
After computing all the pairwise $\delta$-approximate distances, it
remains for each point in $S$ to pick as its $\delta$-approximate nearest
neighbor the point in $S$ within the minimum $\delta$-approximate
distance. This additional step takes $O(n^2)$ time in total.

\begin{theorem}\label{theo: nn}
Let $\delta \in(0,\frac 12)$, and let $S$ be a set of $n$ points in
$\{0,1\}^d.$ One can compute for each point in $S$ a
$\delta$-approximate nearest neighbor in $S$ in time
$O((\log (nd)/\delta^2)
    (nd\log_{1+\delta}d +n^2\log \log_{1+\delta}d ))$.
\end{theorem}

Theorem \ref{theo: nn} cannot be compared with the classical results
on approximate nearest neighbor (e.g., see \cite{AIR18,IM,KOR00}),
which assume a more general situation, where the query
point does not have to belong to the specified subset.

\section{Approximate 0-1 matrix multiplication}

In this section, we combine our procedure for all-pairs Hamming
distances\\
$APPAPHAM$ with our reduction of arithmetic 0-1 matrix
product to the Hamming distance matrix product given in Theorem
\ref{theo: red}(2) in order to approximate the arithmetic 0-1 matrix
product.

The known approximation methods for the arithmetic matrix product tend
to use an aggregated measure to express the goodness of the
approximation, most frequently the Frobenius norm (e.g., see
\cite{CL99,P13}).
We shall provide an estimation of the deviation of each
entry of our approximate matrix product from the correct value in
terms of the Hamming distance between the corresponding row of the
first matrix and the corresponding column of the second matrix. When the distance is
relatively small or on the contrary relatively large, (our upper bound
on) the deviation is small.

\bigskip

    \noindent
        {\bf procedure} $APPROXMM (A,B, \delta )$
        \par
        \noindent
            {\em Input}: Two 0-1 matrices $A$ and $B$  of sizes $p\times q$ and $q\times r,$
            respectively, and a real  $\delta \in (0,\frac 12).$
            \par
            \noindent
                {\em Output}: A $p\times r$ matrix $C'$, where
                for $1\le i\le p$ and $1\le j\le r,$ $C'_{ij}$
                is an approximation of the inner product $C_{ij}$ of
                the $i$-th row $A_{i*}$ of $A$  and the $j$-th column
                $B_{*j}$ of $B.$
                \begin{enumerate}

                \item
                  Call $APPAPHAM (A,B, \delta ).$
                \item
                  Set $\bar{B}$ to the matrix
                  resulting from flipping $1$
                  entries to $0$ and $0$ entries
                  to $1$ in the matrix $B.$
                \item
                  Call $APPAPHAM (A,\bar{B}, \delta ).$
                 \item
                 For $1\le i\le p$,
                 compute the number $A_{i*}^1$ of occurrences of $1$ in $A_{i*}$.
               \item
                 For $1\le j \le r,$
                 compute the number $B_{*j}^1$ of occurrences of $1$ in $B_{*j}$.
               \item
                 For $1\le i\le p$ and $1\le j \le r,$
                 set $D_{ij}$  to
                 $\frac {A_{i*}^1
                   +B_{*j}^1-W_{ij}}2$, where $(W_{ij})$ is the matrix returned
                 by the first call of $APPAPHAM$.
               \item
                For $1\le i\le p$ and $1\le j \le r,$
                 set $D'_{ij}$  to
                 $\frac {A_{i*}^1
                   +B_{*j}^1-(q-W'_{ij})}2$, where $(W'_{ij})$ is the matrix returned
                 by the second call of $APPAPHAM$.
               \item
                 For $1\le i\le p$ and $1\le j \le r,$
                 if $W_{ij}\le W'_{ij}$ then set $C'_{ij}$ to $D_{ij}$ as an approximation
                 of $C_{ij}$ otherwise set $C'_{ij}$ to $D'_{ij}$ as the approximation.
                 
   \end{enumerate}

                The following lemma suggests that the estimations
                $D_{ij}$ might be better when $\mathrm{ham}(A_{i*},B_{*j})$ are
                small while the estimations $D'_{ij}$ might be better
                when the differences $q-\mathrm{ham}(A_{i*},B_{*j})$ are
                small.

\begin{lemma}\label{lem: third}
  For $1\le i \le p$ and $1\le j\le r,$
  the estimations $D_{ij}$ and $D'_{ij}$ of the
  entry $C_{ij}$ of the arithmetic product $C$ of $A$ and $B$ produced
  by the procedure $APPROXMM (A,B, \delta )$ satisfy the following
  inequalities w.h.p.:

  $$ C_{ij}-\frac{\delta\times  \mathrm{ham}(A_{i*},B_{*j})}2 \le D_{ij}
  \le C_{ij}+\frac {\delta \times \mathrm{ham}(A_{i*},B_{*j})} {2+2\delta}   $$
  $$ C_{ij}-\frac {\delta (q-\mathrm{ham}(A_{i*},B_{*j}))}{2+2\delta}\le D'_{ij}  \le C_{ij} +
  \frac {\delta (q-\mathrm{ham}(A_{i*},B_{*j}))}2.$$
\end{lemma}
\begin{proof}
  Let us assume the notation from the proof
  of Theorem \ref{theo: red}. By the aforementioned proof,
  we have $C_{ij}=\frac {A_{i*}^1 + B_{*j}^1-\mathrm{ham}(A_{i*},B_{*j})}2$.
  Hence, by the first inequality stated in Lemma \ref{lem: ineq},
  we obtain $D_{ij}\ge C_{ij}-\delta \mathrm{ham}(A_{i*},B_{*j})/2$ w.h.p.
  Similarly, by the second inequality in Lemma \ref{lem: ineq},
  we obtain $D_{ij}\le  C_{ij}+\delta \mathrm{ham}(A_{i*},B_{*j})/(2+2\delta)   $ w.h.p.
  Analogously, we obtain the inequalities on $D'_{ij}$ by observing
  that $\mathrm{ham}(A_{i*},\bar{B}_{*j})=q-\mathrm{ham}(A_{i*},B_{*j}).$
  \qed
\end{proof}

\begin{theorem}\label{theo: appmm} 
  Let
  $A$ and $B$ be two 0-1 matrices of sizes $p\times q$ and
  $q\times r,$ respectively. For  $\epsilon >0,$
  one can compute an approximation $C''$ of the arithmetic
  matrix product $C$ of $A$ and $B$ 
  such that w.h.p. for $1\le i \le p$ and $1\le j\le r$,
  the approximation $C''_{ij}$ of the inner product $C_{ij}$
  of $A_{i*}$ and $B_{*j}$ 
   differs at most by $\epsilon \min \{\mathrm{ham}(A_{i*},B_{*j}),\ q-\mathrm{ham}(A_{i*},B_{*j})\}$
  from $C_{ij}$
  and   $C''$
  can be computed in time
  $O((\log N /\epsilon^2 )(N\log_{1+\epsilon}q+pr\log \log_{1+\epsilon}q))$,
  where $N$ stands for the input size $pq+qr.$
\end{theorem}
\begin{proof}
  Let us apply $APPROXMM(A,B,\delta)$ for $\delta \in (0,\frac 12)$.
  By Lemma \ref{lem: third}, we have that $D_{ij}$ differs at most by $\frac {\delta}2 \mathrm{ham}(A_{i*},B_{*j})$
  from $C_{ij}$ while $D'_{ij}$ differs at most by $\frac {\delta}2 (q-\mathrm{ham}(A_{i*},B_{*j}))$
  from $C_{ij}.$ It follows also from Lemma \ref{lem: ineq} that
  $W_{ij}$ and $W'_{ij}$ are w.h.p. $\delta$-approximations of $\mathrm{ham}(A_{i*},B_{*j})$ and $q-\mathrm{ham}(A_{i*},B_{*j})$,
  respectively. Hence, $C'_{ij}$ in $APPROXMM$ is based on an $\delta$-approximation
  of $\min \{ \mathrm{ham}(A_{i*},B_{*j}),\ q-\mathrm{ham}(A_{i*},B_{*j})\}$.
  Thus, it is sufficient to pick
  $\delta$ smaller by a constant factor than $\epsilon$
  and apply $APPROXMM(A,B,\delta)$, setting $C''_{ij}$ to $C'_{ij}$
  for $1\le i\le p$ and $1\le j \le r,$  to obtain the theorem.
  \qed
\end{proof}

Note that when $\min \{\mathrm{ham}(A_{i*},B_{*j}),\ q-\mathrm{ham}(A_{i*},B_{*j})\}=O(C_{ij})$
and $\epsilon$ is sufficiently small,
$C''_{ij}$ in Theorem \ref{theo: appmm} is a close approximation of $C_{ij}.$
See also Table 1
for comparison with other approximation methods.

\begin{table*}[t]
\begin{center}
\begin{tabular}{||c|c|c||} \hline \hline
additive error bound & time complexity & reference
\\ \hline \hline
  $\Omega(n/\sqrt c)$ &  $O(n^2c)$ & Drineas et al. \cite{DKM06}
 \\ \hline
$O((\sum_kA_{ik})^{1/2} (\sum_k B_{kj})^{1/2}/\sqrt c)$ & $O(n^2c)$ & Sarl\'os \cite{S06}
\\ \hline
  $||AB||_F/\sqrt c$ & $\tilde{O}(n(n+c))$ & Pagh \cite{P13}, Kutzkov \cite{K13}
  \\ \hline
  $\min \{\mathrm{ham}(A_{i*},B_{*j}),\ n-\mathrm{ham}(A_{i*},B_{*j})\}/\sqrt c$&
    $\tilde{O}(n^2c\log_{1+\frac 1{\sqrt c}} n)$   &               Theorem \ref{theo: appmm}, $\epsilon=\frac 1{\sqrt c}.$
\\ \hline \hline
\end{tabular}
\label{table: one}
\vskip 0.5cm
\caption{Individual bounds on the additive error of each entry in an approximation of the matrix product of $n\times n$ 0-1 matrices $A$ and $B$.}
\end{center}
\end{table*}

\section{An output-sensitive algorithm for an MST in $\Sigma^d$}

Throughout this section, we shall assume that $\Sigma$ is an alphabet
of $O(1)$ size. Furthermore, for a $k\times \ell$ $\Sigma$ matrix $D$,
we shall denote by $M_D$ the cost of a minimum spanning tree of its
rows in the generalized Hamming space $\Sigma^d.$

Our output-sensitive algorithm for a minimum spanning tree in the
generalized Hamming space $\Sigma^d$ is a consequence of the fact that
the Hamming distance matrix product of two $\Sigma$ matrices $A,\ B$
of sizes $p\times q$ and $q\times r$ w.h.p. can be computed in time
$\tilde{O}((p+r)(p+q+r)+\min\{rM_A,\ pM_{B^{\top}}\})$. This fact given in
Theorem \ref{theo: sigmadist} is combined with a standard algorithm
for minimum spanning tree in an edge weighted graph.

In the special case $\Sigma=\{0,1\}$, Theorem \ref{theo: sigmadist}
can be obtained by composing a similar theorem for the arithmetic
product of 0-1 matrices with the reduction given in
Theorem  \ref{theo: red} (1).  Such a similar
theorem can be obtained by a straightforward generalization of Theorem
3.12 in \cite{FJLL18} from square matrices to rectangular ones and the
observation that the geometric distance used in \cite{FJLL18} in case
of 0-1 vectors is just the extended Hamming distance \cite{GL03}
which never exceeds the Hamming one (see also Final remarks).

We provide a direct proof for the general case of $\Sigma$
in terms of the Hamming distances and $\Sigma$
matrices.
\junk{The direct proof shaves off a polylogarithmic factor by
  using Theorem \ref{theo:  approxmst} instead of Fact \ref{fact: mst}}
It relies on Corollary \ref{cor:g gapproxmst} and the following procedure.

\bigskip

\noindent
{\bf procedure $MMST(A,B)$}
\par
\noindent
    {\em Input:} Two $\Sigma$ matrices $A$ and $B$, of sizes $p\times q$ and $q\times r$,
    respectively.
\par
\noindent
    {\em Output:} The Hamming distance matrix product $D$ of $A$ and $B$,
    i.e., for $1\le i\le p$ and $1\le j\le r,$ $D_{ij}$ is the Hamming distance
    between $A_{i*}$ and $B_{*j}.$
\par
\noindent
\begin{enumerate}
\item Construct an $O(1)$-approximate
  spanning tree $T$ of the rows $A_{i*}$ of $A$,
  $i\in [p]$,
  in the generalized Hamming space $\Sigma^q$ by using the method of Corollary \ref{cor:g gapproxmst}.
\item
 Construct a traversal 
(i.e., a non-necessarily simple path visiting all vertices) $U$ of $T.$
\item For any pair $A_{k*}$, $A_{i*}$,
where the latter row follows the former in
the traversal $U,$ compute the set $diff(k,i)$
of indices $\ell \in [q]$, where $A_{i\ell}\neq A_{k\ell}$.
\item
For $j=1,\ldots,r$, iterate the following steps:
\begin{enumerate}
\item Compute $D_{sj}$ where $A_{s*}$ 
is the row of $A$ from which the traversal $U$ of $T$ starts.
\item While following $U$,
iterate the following steps:
\begin{enumerate}
\item
Set $k,\ i$ to the indices
of the previously traversed row of $A$ and
the currently traversed row of $A$, respectively.
\item
Set $D_{ij}$ to $D_{kj}$.
\item For each $\ell \in diff(k,i)$,
  if $A_{i\ell}=B_{\ell j}$ then set
  $D_{ij}$ to $D_{ij}-1$ and if  $A_{k\ell}=B_{\ell j}$ then set
  $D_{ij}$ to $D_{ij}+1$.
\end{enumerate} 
\end{enumerate} 
\end{enumerate}
\junk{
For a $k\times \ell$ matrix 0-1,
let $M_D$ stand for the cost of minimum spanning tree of its rows
in the Hamming space $\{ 0,\ 1\}^{\ell}.$
}

\begin{theorem}\label{theo: sigmadist}
  Let $A,\ B$ be two 
$\Sigma$ matrices of sizes $p\times q$ and $q\times r,$ respectively.
The Hamming distance matrix
product of $A$ and $B$ (i.e., for $1\le i \le p$ and $1\le j \le r,$
the Hamming distance between $A_{i*}$ and $B_{*j}$)
can be computed 
by a randomized algorithm
in time $\tilde{O}((p+r)(p+q+r)+ \min \{rM_A,\ pM_{B^{\top}}\} )$ with high probability.
\end{theorem}
\begin{proof}
  We run the procedure $MMST$ on the input matrices $A,\ B$
  and also on the matrices $B^{\top},\ A^{\top},$ in parallel, and stop whenever
  any of the calls is completed. Note that if $D$
  is the Hamming distance product of $A$ and $B$
  then $D^{\top}$ is the Hamming distance product of $B^{\top}$ and $A^{\top}.$
  The correctness of the procedure follows from the correctness of
  the updates of $D_{ij}$ in the block of the inner loop, i.e.,
  in Step 4(b).

  Step 1 of $MMST(A,B)$ takes
  $O(\log (pq) (pq\log q  +p^2\log \log q ))$
  time by Corollary \ref{cor:g gapproxmst}.
  Step 2 can be done in $O(p)$ time while Step 3 requires
  $O(pq)$ time. The first step in the block under
  the outer loop, i.e., computing
  $D_{sj}$ in Step 4(a), takes $O(q)$ time.
  The crucial observation is that the second step
  in this block, i.e., Step 4(b), requires $O(p+M_A)$ time.
  Simply, the substeps (i), (ii) take $O(1)$ time while
  the substep (iii) requires $O(|diff(k,i)|+1)$ time.
  Since the block is iterated $r$ times, the whole
  outer loop, i.e., Step 4, requires $O(qr+pr+rM_A)$ time.
  Thus, $MMST(A,B)$ can be implemented in time
  $O(\log (pq) (pq\log q  +p^2\log \log q )+qr+pr+rM_A)$
 
  Symmetrically, $MMST(B^{\top},A^{\top})$ can be done in time
  $O(\log (qr) (qr\log q  +r^2\log \log q )+pq+pr+pM_{B^{\top}})$.
  \qed
  \end{proof}
\junk{
\begin{fact} \label{fact: first}
Let $A,\ B$ be two 
0-1 matrices of sizes $p\times q$ and $q\times r,$ respectively.
The arithmetic matrix
product of $A$ and $B$
can be computed 
by a randomized algorithm
in $\tilde{O}(pq+qr+ pr+ \min \{rM_A,\ pM_{B^{\top}}\} )$ time with high probability.
\end{fact}}

\begin{corollary}
  Let $S$ be a set of $n$ points in the generalized Hamming space $\Sigma^d$ and
  let $M$ be the cost of a minimum spanning tree of $S$
  in $\Sigma^d$.
  A minimum spanning tree of $S$ can be computed by a randomized algorithm
  in $\tilde{O}(n(d+n+M))$ time with high probability.
\end{corollary}
\begin{proof}
  \junk{Similarly as in Section 3, form a $m\times d$ matrix
  $P$, where for $i\in {m},$ the $i$-th row is the vector
  of coordinates of the $i$-th input point.
  By Theorem \ref{theo: allpairsham},
  the problem of computing the Hamming distances between all pairs
  of points in $S$ reduces in $O(md)$ to that of computing
  the arithmetic matrix product $PP^{\top}$ and $P\bar{P}^{\top},$
  where $\bar{P}$ is obtained from $P$ by flipping $0$ entries
  to $1$ and {\rm vice versa}.}
  We can compute the Hamming distances between all pairs
  of points in $S$ in   $\tilde{O}(n(d+n+M))$ time
  by Theorem \ref{theo: sigmadist}.
    Then, it is sufficient to run a standard algorithm
  for minimum spanning tree on vertices $1$ through $n$
  corresponding to the points in $S$ and edges weighted
  by the Hamming distances between the points in $S$
  corresponding to the edge endpoints. This final step takes
  $O(n^2)$ time.
  \qed
\end{proof}

For example, if $n=d,$ our output-sensitive algorithm
for a minimum spanning tree is substantially faster than that
of Corollary \ref{cor: gmst} when $M$ is substantially smaller
than $n^{\omega -1}.$

\section{Approximate $\ell$-clustering in high dimensional
  Hamming spaces}

The $\ell$-center clustering problem in a Hamming space
 $\{0,1\}^d$ is as follows: given a set $P$ of $n$ points
 in  $\{0,1\}^d$, find a set $T$ of $\ell$ points in  $\{0,1\}^d$
 that minimize $\max_{v\in P} \min_{u\in T} \mathrm{ham}(v,u).$

The minimum-diameter $\ell$-clustering problem in a Hamming space
 $\{0,1\}^d$ is as follows: given a set $P$ of $n$ points
 in  $\{0,1\}^d$, find a partition of $P$ into $\ell$ subsets
 $P_1,P_2,\dots, P_{\ell}$ that minimize
 $\max_{i\in [\ell]} \max_{v,u \in P_i} \mathrm{ham}(v,u).$

 The $\ell$-center clustering problem could be
 also termed as the minimum-radius $\ell$-clustering problem.
 Both problems are known to be NP-hard
 to approximate within $2-\epsilon$ 
 in metric spaces \cite{Gon85,HS85}.
 \junk{Gonzalez provided a simple $2$-approximation method
 for $\ell$-center clustering that yields also a $2$-approximation
 for minimum-diameter $\ell$-clustering \cite{Gon85}.}
Gonzalez' simple $2$-approximation method for
$\ell$-center clustering yields also a $2$-approximation
 for minimum-diameter $\ell$-clustering \cite{Gon85}.
 It picks an arbitrary input point as the first
 center and repeatedly extends the current center set 
 by an input point that maximizes the distance
 to the current center set, until $\ell$ centers are found.
 In case of $\{0,1\}^d$ Hamming space, his method takes
 $O(nd\ell)$ time, where $n$ is the number of input points.
\junk{
 On a word RAM with computer words of length $w$, the factor $w$ can
 be shaved off. For low dimensional Euclidean spaces of bounded
 dimension, and more generally, for metric spaces of bounded doubling
 dimension, there exist faster $2$-approximation algorithms for the
 $\ell$-center problem with hidden exponential dependence on the
 dimension in their running time, see \cite{FG88} and \cite{HM06},
 respectively.}
   
 By forming for each of the $\ell$ centers, the cluster
consisting of all input points for which this center
is the closest one (with ties solved arbitrarily),
one obtains an $\ell$-clustering with the maximum
cluster diameter within two of the minimum \cite{Gon85}.

In this section, we provide a faster $(2+\epsilon)$-approximation
method for both problems in Hamming spaces of super-logarithmic dimension
when $\ell=\omega (\log n)$ that combines Gonzalez' method
with the randomized dimension reduction described in
Lemma \ref{lem: R}.
\bigskip

    \noindent
        {\bf procedure} $CENTER(\ell, P, \delta )$
        \par
        \noindent
            {\em Input}: A positive integer $\ell,$
a set $P$ of points $p_1,\dots,p_n \in \{0,1\}^d,$ $n>\ell,$
and a real $\delta \in (0,\frac 12).$
            \par
            \noindent
        {\em Output}: An $\ell$-center clustering $T$ of $P.$

                \begin{enumerate}
\item
                  Set $n$ to the number of input points
                 and $k$ to $O(\log n /\delta^2)$.
                 \item
                 Generate random $k\times d$     0-1 matrices 
  $F_t$ for the thresholds $t=1,\ (1+\delta),\dots, (1+\delta)^{\lceil \log _{1+\delta} d\rceil },$
                  defining the functions $f_t$ by $f_t(x)=F_tx^{\top}\mod 2$
                  (see Remark \ref{rem: 1}).
                \item
                  Compute the values of all the functions $f_t$  for each
                  point $p_i\in P,$
                  i.e.,
                  for $i \in [n],$
                  compute $F_tp_i^{\top}\mod 2$ for
  $t=1,\ (1+\delta),\dots, (1+\delta)^{\lceil  \log _{1+\delta}d \rceil}.$
\item
Set $T$ to $\{ p_1 \}$, and for $j\in [n] \setminus \{ 1 \}$,
  set $W_{1j}$  to the smallest $t\in \{
1,\ (1+\delta),\dots, (1+\delta)^{\lceil \log _{1+\delta}d\rceil } \}$
such that $\mathrm{ham}(f_t(p_1),f_t(p_j))\le (C_1+\frac {\delta}
{30})k$ (see Lemma \ref{lem: R}).
\item
$\ell -2$ times iterate the following three steps:
\begin{enumerate}
\item
Find $p_m \in P\setminus T$ that maximizes $\min_{p_q\in T}
W_{qm}$ and extend $T$ to $T\cup \{p_m\}.$
\item
For each $p_j\in P\setminus T,$ set $W_{mj}$
to the smallest $t\in \{
1,\ (1+\delta),\dots, (1+\delta)^{\lceil \log _{1+\delta}d\rceil } \}$
such that $\mathrm{ham}(f_t(p_m),f_t(p_j))\le (C_1+\frac {\delta}
{30})k.$
\item
For each $p_j\in P\setminus T,$ update $\min_{p_i \in T} W_{ij}.$
\end{enumerate}
\item
Find $p_m \in P\setminus T$ that maximizes $\min_{p_q\in T}
W_{qm}$ and extend $T$ to $T\cup \{p_m\}.$
\item
Return $T.$

 \end{enumerate}

  \begin{lemma}\label{lem: timecenter}
    $CENTER(\ell,P,\delta)$ runs in time\\
    $O(n \log n (d \log_{1+\delta} d + \ell \log
\log_{1+\delta} d)/\delta^2)$.
\end{lemma}
\begin{proof}
  Step 1 can be done $O(nd)$ time. Step 2 takes $O(dk\log_{1+\delta} d)\
  =$\\
$O((d\log n \log_{1+\delta} d)/\delta^2)$
  time. The preprocessing in Step 3 requires \\
$O(ndk\log_{1+\delta} d)$, i.e., 
$O((nd\log n \log_{1+\delta} d)/\delta^2)$ time.
Step 4 can be done in\\
 $O((n\log n \log \log_{1+\delta}d)/\delta^2 )$ time by binary
search. Steps 5(a) and 5(c) take $O(n)$ time while Step 5(b) as
Step 4 can be done in $O((n\log n \log \log_{1+\delta}d) /\delta^2)$ time by binary search.
Consequently, the whole Step 5 requires $O((n\ell \log n
\log \log_{1+\delta}d)/\delta^2 )$
time. Finally, Step 6 takes $O((n\log n \log \log_{1+\delta}d )/\delta^2)$ time
similarly as Steps 4 and 5(b). It remains to observe that 
the overall running time of $CENTER(\ell,P,\delta)$
is dominated by those of Step 3 and Step 5.
 \end{proof}

 \begin{theorem}
   Let $P$ be a set of $n$ points $p_1,\dots,p_n \in \{0,1\}^d,$
   $\ell$ an integer smaller than $n,$ and let $\epsilon\in (0,1/2).$
   $CENTER(\ell,P,\epsilon/5)$
   provides a $2+\epsilon$
   approximation of an optimal $\ell$-center clustering in time
    $O(n \log n (d \log_{1+\epsilon/5} d + \ell \log
\log_{1+\epsilon/5} d)/\epsilon^2).$
   It also yields a
   $2+\epsilon $ approximation to an $\ell$-clustering of $P$ with
   minimum cluster diameter.
 \end{theorem}
 \begin{proof}

   Let $\delta=\epsilon/5$ and let $T$ be the set of
   $\ell$ centers returned by $CENTER(\ell,P,\delta)$.
   Next, let $r=\max_{v\in S} \min_{u\in T} \mathrm{ham}(v,u)$,
   $r_w=\max_{p_i\in P}\min_{p_j\in T} W_{ij}$, and let $p_q$
 be a point for which the latter maximum is achieved.
 Since $W_{ij}$ are defined analogously as
 in the procedure $APPAPHAM$,
 it follows from Lemma \ref{lem: ineq} that $r_w\ge r/(1+\delta).$
 By the specification of the
 procedure $CENTER$ and the definition of $p_q,$
 the set $T\cup \{p_q\}$ consists of $\ell+1$ points such that for any pair
 $p_v,\ p_u$ of points in this set $W_{vu}\ge r_w$ holds.
 Consequently, these $\ell+1$
 points are at the Hamming distance at least $r_w/(1+\delta)$
 apart by Lemma \ref{lem: ineq}. Let $T^*$ be an optimal set of $\ell$ center
 points.
 Two of these $\ell+1$ points in  $T\cup \{p_q\}$
 must have the same nearest center in $T^*$.
 It follows that the Hamming distance of at least one of these two points
 to its nearest center
 in $T^*$ is at least $\frac {r_w}{2(1+\delta)}.$
 Since $r_w\ge r/(1+\delta)$ by Lemma \ref{lem: ineq},
 we infer that $\max_{p_i\in P} \min_{p_j\in T^*} \mathrm{ham}(p_i,p_j)$
 is at least  $\frac {r}{2(1+\delta)^2}.$
 It remains to note that $2(1+\delta)^2\le 2+\epsilon$
 by $\delta =\epsilon /5.$
 This combined with Lemma \ref{lem: timecenter}
completes the proof of the first part.

We can slightly modify  $CENTER(\ell,P,\delta)$
so it returns a partition of $P$
into clusters $P_i$, $i\in [\ell],$ each consisting of all points
in $P$ for which the $i$-th center is closest in terms
of the approximate $W_{ij}$ distances. Note that the maximum
diameter of a cluster in this partition is at most
$2r_w(1+\delta).$ To implement the modification, we let to
update in Step 5(c) not only $\min_{p_i \in T} W_{ij}$ but
also the current center $p_i$ minimizing  $\min_{p_i \in T} W_{ij}$.
This slight modification does not alter the asymptotic complexity
of  $CENTER(\ell,P,\delta)$.

The proof of the second part follows by a similar argumentation.
Consider the $\ell +1$ points defined in the proof of the first part.
Recall that they are at least  $r_w/(1+\delta)$ apart.
 Two of the $\ell+1$ points have also to belong to the same cluster
 in a $\ell$-clustering that minimizes the diameter. It follows
 that the minimum diameter is at least $r_w/(1+\delta)$
 while the diameter of the clusters in the $\ell$-center
 clustering returned by $CLUSTER(\ell,P,\delta)$
 is at most  $2r_w(1+\delta)$ by Lemma \ref{lem: ineq}.
 Consequently, it is larger by at most $\frac {2r_w(1+\delta)}
 {r_w/(1+\delta)} \le 2 +\epsilon$ than the optimum.
 \qed
\end{proof}

\section{Final remarks}

In order to extend also the results of Section 4 to include the
generalized Hamming space $\Sigma^d$ one needs to extend Lemma
\ref{lem: R}. For this purpose, presumably one can use the ideas of
the corresponding extension of the results on approximate nearest
neighbors search in Hamming spaces outlined in Section 4 in \cite{KOR00}.

The extended Hamming distance ($eh(\ ,\ )$)
                     introduced in \cite{GL03} is a generalization of
                     the Hamming distance to include blocks of zeros
                     and ones. It is defined recursively for two 0-1
                     strings (alternatively, vectors or sequences)
                     $s=s_1 s_2\dots s_m$ and $u=u_1 u_2\dots u_m$ as
                     follows:
                     $$eh(s,u)=eh(s_{\ell +1}\dots s_m,u_{\ell +1}\dots u_m)+
                     (s_1+u_1)\mod 2,$$
                     where $\ell$ is the maximum number such that
                     $s_j=s_1$ and $u_j=u_1,$ for $j=1,\dots,\ell.$
                     Note that the extended Hamming distance between
                     any two strings or vectors never exceeds the
                     Hamming one. The authors of \cite{GL03} provide
                     a linear transformation of the input 0-1 vectors that reduces
                     the computation of 
                     the extended Hamming distance
                     between them to the computation of the Hamming distance
                     between the transformed vectors. Hence, our results
                     in terms of the Hamming distance in hypercubes can be generalized
                     to include the extended Hamming distance.
\junk{
                     Another line of extension of our results is to
                     consider the generalized Hamming space
                     $\Sigma^d$, where $\Sigma$ is fined alphabet
                     \cite{A18,KOR00}. In fact, Arslan reduced the
                     problem of computing all-pairs Hamming distances
                     in $\Sigma^d$ to $O(|\Sigma|)$ arithmetic products of 0-1
                     matrices in \cite{A18}, so Section 3 can be directly
                     extended to include the generalized Hamming
                     space. In case of Section 4 and consequently
                     Section 6, one needs to extend Fact \ref{fact:
                       R}. For this purpose, one can use  ideas of
                     corresponding extension of results on approximate
                     nearest neighbors search in Hamming spaces outlined
                     in Section 4 in \cite{KOR00}.}
                     
     \small
  \bibliographystyle{abbrv}                    
\bibliography{Voronoi3}
\appendix
\section{Proof of Lemma \ref{lem: R}}
\noindent
Lemma \ref{lem: R} is similar to Lemma 2.3 in \cite{AIR18} and as that
based on arguments from \cite{KOR00}.
Unfortunately, no proof
of Lemma 2.3 in \cite{AIR18} seems to be available in the literature.

\begin{proof}
To define the map $f,$ generate a random $k\times d$ 0-1 matrix $F,$
  where each entry is set to $1$
  with probability $\frac 1 {4t}$. For $x\in \{ 0,1\}^d,$ let $f(x)=Fx^{\top}$ over the field $GF(2).$

  For $x,y \in \{0,1\}^d,$  consider the $i$-th coordinates $f^i(x),\ f^i(y)$ of $f(x)$ and $f(y),$
  respectively. Suppose that $\mathrm{ham}(x,y)=D.$

  To estimate the probability that $f^i(x)\neq f^i(y)$
  when $D>0,$ observe that for $z\in \{0,1\}^d,$ $f^i(z)$ can be equivalently obtained as follows.
 Pick a  random subset $S$ of the $d$ coordinates such that each coordinate is selected with probability
 $\frac 1 {2t}.$ Let $\bar{z}$ be the vector  in $\{0,1\}^d$ obtained by setting all coordinates
 of $z$ not belonging to $S$ to zero. Now the inner mod 2 product of $\bar{z}$ with
 a vector  $r\in \{0,1\}^d$ picked  uniformly at random yields $f^i(z).$

 It follows that the probability that $f^i(x)\neq f^i(y)$ is $\frac
 12(1-(1-\frac 1 {2t})^D).$ Simply,
 strictly following \cite{KOR00}, if none of
 the $D$ coordinates $j,$ where $x_j\neq y_j,$ is chosen into
 $S$ then $f^i(x)=f^i(y),$ otherwise at least one of such coordinates,
 say $m$, is in $S$, and for each setting of other choices,there is
 exactly one of the two choices for $r_m$ which yields $f^i(x)\neq f^i(y).$
 Observe that the probability is increasing in $D.$

  Consequently, if $\mathrm{ham}(x,y)\le t$ then the probability that
 $f^i(x)\neq f^i(y)$ does not exceed $\frac 12(1-(1-\frac 1 {2t})^t)\approx \frac 12 (1-e^{-1/2})$
 while when $\mathrm{ham}(x,y)\ge (1+\epsilon)t$ it is at least
 $\frac 12(1-(1-\frac 1{2t})^{(1+\epsilon)t})\approx 
\frac 12 (1-e^{-(1+\epsilon)/2})$.
 We set $C_1= \frac 12(1-(1-\frac 1 {2t})^t)$
 and $C_2=\frac 12 (1-(1-\frac 1{2t})^{(1+\epsilon)t}).$
  Hence, in the first case the expected value of $\mathrm{ham}(f(x),f(y))$ is
  at most $C_1k$ while in the second case it is at least $C_2k$.
  
Let us estimate
$C_2-C_1.$ By the monotonicity of the function
$(1-\frac{1}{2t})^{t}$ in $t \geq 1$, $C_2-C_1$ is at least
$$\frac{1}{2}\left[ \left( 1-\frac{1}{2t}\right)^{t}-
\left(1-\frac{1}{2t}\right)^{(1+\epsilon)t}\right] =
\frac{1}{2}\left(1-\frac{1}{2t}\right)^{t}\left(1-\left(1-\frac{1}
  {2t}\right)^{\epsilon t}\right) \geq
\frac{1}{4}(1-e^{-\frac{\epsilon}{2}}).$$
We expand $e^{-x}$ in the Taylor
series $\sum_{i=0}^{\infty}\frac{(-x)^i}{i!}$.
Since the sums of consecutive
pairs of odd and even components
are negative, we obtain:
$$C_2-C_1 \geq \frac{1}{4}\left(1-e^{-\frac{\epsilon}{2}}\right) \geq
\frac{1}{4}\left(1-\left(1-\frac{\epsilon}{2}+\frac{(\frac{\epsilon}{2})^2}{2}\right)\right) =
\frac{1}{4}\left(\frac{\epsilon}{2}-\frac{\epsilon^2}{8}\right).$$
Hence, since $\epsilon < \frac{1}{2}$, we infer:
$$C_2-C_1 \geq \frac{1}{4}\left(\frac{\epsilon}{2}-\frac{\epsilon^2}{8}\right) \geq
\frac{1}{4}\left(\frac{\epsilon}{2}-\frac{\epsilon}{16}\right)
= \frac{7 \epsilon}{64}\ge \frac  {\epsilon}{10}.$$

 We shall use the following 
Chernoff multiplicative bounds (see Appendix A in \cite{AS}).  For a
sequence of independently and identically distributed (i.i.d.) 0-1
random variables $X_1,X_2,\dots,X_m$, $Pr[\sum X_i >(p+\gamma)m] <
e^{-pm\gamma^2}$, and $Pr[\sum X_i <(p-\gamma)m] < e^{-2m\gamma^2},$
where $p=Pr[X_i=1]$.

By the Chernoff bounds,
if $\mathrm{ham}(x,y)\le t$ then the probability
 that $\mathrm{ham}(f(x),f(y))>  (C_1+\frac 1{30} \epsilon)k$
 is at most $e^{ -\frac 1{900} \epsilon^2 k}$.
 Similarly, if $\mathrm{ham}(x,y)\ge (1+\epsilon)t$
 then the probability that $\mathrm{ham}(f(x),f(y)) < (C_2-\frac 1{30} \epsilon)k$
 is at most $e^{ -\frac 2{900} \epsilon^2 k}$.
 
 Note that $C_2-\epsilon/30 > C_1+\epsilon/30$ by $C_2-C_1\ge  \frac  {\epsilon}{10}.$
 \qed
\end{proof}
\end{document}